\documentclass[english,journal]{extarticle}
\usepackage{fontenc}
\usepackage[utf8]{inputenc}
\usepackage{color}
\usepackage{subfigure}
\usepackage{graphicx}
\usepackage{babel}
\usepackage{float}
\usepackage{amsthm}
\usepackage{amsmath}
\usepackage{amssymb}
\usepackage{tikz}
\usepackage{fixltx2e}
\usepackage{multirow}
\usepackage{anysize}
\usepackage{parskip}
\usepackage{xargs}[2008/03/08]
\usepackage[unicode=true,pdfusetitle,
 bookmarks=true,bookmarksnumbered=false,bookmarksopen=false,
 breaklinks=false,pdfborder={0 0 0},backref=false,colorlinks=true]
 {hyperref}
\marginsize{1in}{1in}{1in}{1in}

\makeatletter

\usepackage{xy}
\xyoption{matrix}
\xyoption{frame}
\xyoption{arrow}
\xyoption{arc}

\usepackage{ifpdf}
\ifpdf
\else
\PackageWarningNoLine{Qcircuit}{Qcircuit is loading in Postscript mode.  The Xy-pic options ps and dvips will be loaded.  If you wish to use other Postscript drivers for Xy-pic, you must modify the code in Qcircuit.tex}
\xyoption{ps}
\xyoption{dvips}
\fi

\entrymodifiers={!C\entrybox}

\newcommand{\ket}[1]{{\left\vert{#1}\right\rangle}}
\newcommand{\qw}[1][-1]{\ar @{-} [0,#1]}
\newcommand{\qwx}[1][-1]{\ar @{-} [#1,0]}

\newcommand{\gate}[1]{*+<.6em>{#1} \POS ="i","i"+UR;"i"+UL **\dir{-};"i"+DL **\dir{-};"i"+DR **\dir{-};"i"+UR **\dir{-},"i" \qw}

\newcommand{\control}{*!<0em,.025em>-=-<.2em>{\bullet}}

\newcommand{\ctrl}[1]{\control \qwx[#1] \qw}

\newcommand{\targ}{*+<.02em,.02em>{\xy ="i","i"-<.39em,0em>;"i"+<.39em,0em> **\dir{-}, "i"-<0em,.39em>;"i"+<0em,.39em> **\dir{-},"i"*\xycircle<.4em>{} \endxy} \qw}

\newcommand{\multigate}[2]{*+<1em,.9em>{\hphantom{#2}} \POS [0,0]="i",[0,0].[#1,0]="e",!C *{#2},"e"+UR;"e"+UL **\dir{-};"e"+DL **\dir{-};"e"+DR **\dir{-};"e"+UR **\dir{-},"i" \qw}
\newcommand{\ghost}[1]{*+<1em,.9em>{\hphantom{#1}} \qw}

\newcommand{\lstick}[1]{*!R!<.5em,0em>=<0em>{#1}}

\newcommand{\Qcircuit}{\xymatrix @*=<0em>}

\theoremstyle{plain}

\theoremstyle{definition}
\newtheorem{defin}{Definition}

\theoremstyle{remark}

\newtheorem{lemm}{Lemma}
\newenvironment{lemma}{\vspace{0mm}\begin{lemm}}{\end{lemm}}
\renewenvironment{proof}{\noindent{\bf Proof:}\;}{$\square$\,}

\begin{document}

\title{Basic circuit compilation techniques for an ion-trap quantum machine}

\author{\small{ Dmitri Maslov$^{1,2}$} \\
{\small\it $^1$ National Science Foundation, Arlington, VA, USA} \\
{\small\it $^2$ QuICS, University of Maryland, College Park, MD, USA} \\
{\small\tt \href{mailto:dmitri.maslov@gmail.com}{dmitri.maslov@gmail.com}}\\
}

\maketitle

\begin{abstract}
We study the problem of compilation of quantum algorithms into optimized physical-level circuits executable in a quantum information processing (QIP) experiment based on trapped atomic ions. We report a complete strategy: starting with an algorithm in the form of a quantum computer program, we compile it into a high-level logical circuit that goes through multiple stages of decomposition into progressively lower-level circuits until we reach the physical execution-level specification.  We skip the fault-tolerance layer, as it is not within the scope of this work.  The different stages are structured so as to best assist with the overall optimization while taking into account numerous optimization criteria, including minimizing the number of expensive two-qubit gates, minimizing the number of less expensive single-qubit gates, optimizing the runtime, minimizing the overall circuit error, and optimizing classical control sequences.  Our approach allows a trade-off between circuit runtime and quantum error, as well as to accommodate future changes in the optimization criteria that may likely arise as a result of the anticipated improvements in the physical-level control of the experiment.
\end{abstract}


\section{Introduction}
The interest in quantum computing is rooted in the ability to solve certain computational problems more efficiently by a quantum algorithm than it is known how to do by a regular classical algorithm \cite{bk:nc}.  To take advantage of those quantum algorithms, a suitable quantum information processing (QIP) system needs to be developed---specifically, one that provides the means to efficiently execute protocols prescribed by the respective quantum algorithms \cite{ar:div}.  As of the time of this writing, fully programmable quantum computational devices spanning a few to several qubits included those built based on the superconducting circuits \cite{ar:dh, www:IBM} and trapped ions \cite{ar:blawi, ar:deb, ar:ladd} technologies.

Since the focus of this paper is on the computing over trapped ions QIP platform, we next quickly describe how it works.  For details specific to this paper, also see \cite{ar:deb}.  In the trapped ions QIP the qubits are stored in the spins of the individual ions ($^{171}$Yb$^+$ in \cite{ar:deb}), with the ions suspended in the free space (vacuum) via the use of electromagnetic fields.  When confined in two dimensions, ions form a line, spanning the remaining spacial dimension.  Weak confinement in the third dimension can maintain a linear structure of the ion crystal.  Observed qubit coherence time of $0.5s$ is so long that it is currently not a limiting factor on the size of the computation that is possible to execute; furthermore, it is expected that it can be scaled up by the orders of magnitude in the future \cite{ar:fslc}.  Lasers are used to both initialize the state of the system to a simple state $\ket{00...0}$ via a process called optical pumping, and to read out the state, relying on the fluorescence---specifically, through applying a laser that couples to only one of the two qubit states, and as a result emitting a stream of fluorecent photons.  Both state initialization and measurement are implemented with a very high efficiency---for all practical purposes, that approaching a 100\% accuracy.  Single- and two-qubit gates are implemented via laser-driven stimulated Raman transitions.  This gives rise to the single-qubit physical level gates $R$ and two-qubit $XX$ interaction discussed in detail later in this paper; these gates can be applied to an arbitrary and selectable set of qubits, and form a computationally universal gate library.  Gate fidelities are high, with the demonstrated average CNOT fidelity of $95.6-98.5\%$, depending on the pair of qubits the respective CNOT is being applied to \cite{ar:deb}.  Note that since the directly implementable elementary gates are $R$ and $XX$, the CNOT gate itself is a composite transformation.  Specifically, \cite{ar:deb} uses a 1-$XX$ and 6-$R$ implementation of the CNOT gate\footnote{Observe that this paper introduces a 1-$XX$ and 4-$R$ implementation of the CNOT gate, thereby likely improving the CNOT fidelities reported in \cite{ar:deb}.  The extent to which our optimized implementation improves over the one reported in \cite{ar:deb} needs to be established through the experiment.}.  The fidelities of the native physical-level $R$ and $XX$ gates are higher than that of the CNOT; we also note that the single-qubit $R$ gates have considerably higher fidelity compared to the $XX$ gates. The authors of \cite{ar:deb} furthermore expect to scale physical-level gate fidelities to $99.9\%$ and above with future improvements to the classical control hardware. 

Control over systems of several qubits and their interactions has reached a level where quantum algorithms consisting of many dozens of physical gates ({\em e.g.}, 80 physical single- and two-qubit gates in the QFT5 experiment demonstrated in \cite{ar:deb}) are within the reach and circuit optimization becomes a crucial part of their realization.  Here we report basic gate decompositions and a general circuit design and optimization approach that can be applied to systematically assemble relevant computational experiments. 

We will work with the pure quantum $n$-qubit states as given by the state vector $\sum_{i=0}^{2^n-1}\alpha_i\ket{i}$ and quantum circuits, defined as the products of quantum gates.  A quantum gate over a set of $n$ qubits is described by a $2^n \times 2^n$ unitary matrix $U$.  This formalizes the mathematical properties of the transformations that are possible in principle, but does not specify which of those can be implemented directly on the physical level, or how to compose those physical-level gates into efficient circuit sequences.  What is and is not possible to obtain on the physical level furthermore depends on the choice of the QIP platform and the available controlling apparatus.  In this paper we focus on the trapped ions approach \cite{ar:deb}.  We first describe physical-level gates obtained in the lab experiment, and show how to use them to efficiently implement known popular logical-level quantum gates such as Pauli gates and their roots, Clifford gates, the CNOT, controlled roots of Paulis, and the Toffoli gate---constituting a set most often used when describing quantum algorithms.  Note that due to full qubit connectivity, quantum SWAP gates may be pushed to the end of the quantum circuit and thereby implemented classically at no cost to the respective quantum computation.  Next, we propose a generic optimizing compiler that maps logical-level quantum circuits into efficient physical experiments.   We conclude the paper with benchmark results showing how the techniques developed can be applied to design optimized quantum computational experiments larger than those demonstrated previously, yet suitable for execution on the existing hardware \cite{ar:deb}.

\section{Single-Qubit Gates}

\subsection{Physical-level single-qubit rotation}

The controlling apparatus allows the application of the single-qubit rotation $R(\theta,\phi)$ described by the following unitary evolution operator:
\begin{equation*}
R(\theta,\phi):=\left( \begin{array}{rr}
\cos\frac{\theta}{2} & -ie^{-i\phi}\sin{\frac{\theta}{2}}\\
-ie^{i\phi}\sin\frac{\theta}{2} & \cos\frac{\theta}{2}
\end{array} \right).
\end{equation*}
Both $\theta$ and $\phi$ can be controlled by changing the duration and phase of the Raman beatnote that drives the Rabi oscillation of the qubit \cite{ar:deb}.

\subsubsection{Single-qubit gate cost}

The gate $R(\theta,\phi)$ has two cost parameters,
\begin{eqnarray}\label{eqn:ed}
d:=\frac{|\theta|}{\pi}\tau_{1q}, \text{ and }
e=e_1:=|\sin(\theta)|\epsilon \text{ or } e=e_2:=(|\theta| \bmod \pi)\epsilon,
\end{eqnarray}
where $d$ is the duration of the above single-qubit rotation and $e$ gives a model of the experimental error based on laser pulse area fluctuations due to laser intensity and timing jitter, leading to random over-/under-rotations of the qubit.  The formula describing $e_1$ is constructed such as to highlight the effect of the slope of the Rabi oscillation being smallest for full $\pi$ rotations when the gate is applied to a quantum state close to the computational basis states \cite{pers}.  For an unknown qubit state, it is impossible to tell the slope of Rabi oscillation, and thereby $e_1$ error model becomes inaccurate.  The formula describing $e_2$ is designed such as to highlight that the error should be proportional to the rotation angle \cite{ar:bskw}.  However, $e_2$ has its own limitations.  Indeed, such a definition predicts that, to consider a specific example, the error in the single pulse circuit $R(\pi/2,0)$ will be smaller than that in the gate $R(\pi,0)$ (both applied to a computational basis state).  However, the experiment shows the opposite result---$R(\pi,0)$ is more accurate than $R(\pi/2,0)$, and thereby $e_2$ is also inaccurate.

We note that while proper explanation and accurate modeling of experimental errors is very important, narrowing down a complete error model is not the focus of this paper.  For the purpose of this paper, any error model can be acceptable.  This is because the goal is to illustrate that trapped ions experiments can be optimized across a combination of (two) conflicting optimization criteria by those techniques reported, and to highlight the inner workings of such optimization approach.  With this in mind, we select $e_1$ model, since optimization over $e_2$ is equivalent to optimization of the duration, reducing the number of optimization criteria to just one. 

Presently, single-qubit rotations, as well as the two-qubit gates are implemented serially.  As a result, the overall runtime of a computational experiment, as described by its circuit, equals to the sum of the runtimes of the individual gates.  Depending on the desired properties of the circuit, one may choose to optimize the overall runtime, the overall error, the overall number of gates (including keeping separate counts of the single-qubit and two-qubit gates), as well as any combined figure of the above.  In this paper, we will describe the overall cost of an implementation as a length-2 vector $(d,e)$, with the components corresponding to the overall duration and the overall error.  The error component itself is described by the list (written as a linear combination) of all errors from all gates participating in the respective circuit, per the error model introduced for the individual gates.  This definition of the error does not correspond to the actual error seen in the experiment, but rather shows the influence and sources or errors within the given implementation.  We try to minimize the cost vector $(d,e)$, focusing separately on the duration and error.  One may choose to focus on other optimization criteria, such as, {\em e.g.}, minimization of the gate count; this does not affect the overall optimization strategy or the steps taken to arrive at the optimized solution. 

For future discussions, we will need the following relations: 
\begin{itemize}
\item $R^{-1}(\theta,\phi) = R(\theta,\phi-\pi)$, that can be used to construct the inverse of the $R(\theta,\phi)$ gate at the same cost as the original gate; and
\item $R(\theta,\phi) = (-1)\cdot R(\theta-2\pi,\phi)$, that can be helpful in that it provides the means for limiting the duration of any one $R(\theta,\phi)$ gate to at most $\tau_{1q}$, as the global phase does not matter.
\end{itemize}
Both identities are easy to verify directly.

\subsection{The RX, RY, and RZ rotations}
The single-qubit rotations around the basis axes must be expressed in terms of the physical-level $R$ gate to be implementable in an experiment.

\vspace{2mm}\noindent {\bf RX:} Setting $\phi= 0$ in $R(\theta,\phi)$ achieves the rotation about the X axis by the angle $\theta$, as follows:
\begin{equation}\label{eq:rx}
RX(\theta):=\left( \begin{array}{rr}
\cos\frac{\theta}{2} & -i\sin\frac{\theta}{2}\\
-i\sin\frac{\theta}{2} & \cos\frac{\theta}{2}
\end{array} \right)
=R(\theta,0).
\end{equation}
Observe that the duration of $RX(\theta)$ is $\frac{|\theta|\tau_{1q}}{\pi}$, whereas its error is $|\sin\theta|\epsilon$, {\it i.e.}, its cost vector is $\left(\frac{|\theta|\tau_{1q}}{\pi}, |\sin\theta|\epsilon\right)$.  The implementation (\ref{eq:rx}) is well known. 

\vspace{2mm}\noindent {\bf RY:} Setting $\phi= \frac{\pi}{2}$ in $R(\theta,\phi)$ obtains the rotation about Y axis by the angle $\theta$.  In particular,
\begin{equation}\label{eq:ry}
RY(\theta):=\begin{pmatrix}
\cos\frac{\theta}{2} & -\sin\frac{\theta}{2}\\
\sin\frac{\theta}{2} & \phantom{-}\cos\frac{\theta}{2}
\end{pmatrix}
= R(\theta, \pi/2).
\end{equation}
As a result, the cost of $RY(\theta)$ is $\left( \frac{|\theta|\tau_{1q}}{\pi}, |\sin\theta|\epsilon \right)$.  The implementation (\ref{eq:ry}) is also well known.  The costs of the $RX$ and $RY$ gates with the same rotation angle are thus the same.  $RZ$ is more difficult to obtain. In particular,

\vspace{2mm}\noindent {\bf RZ:} $RZ$ rotation is defined as follows, 
\begin{equation*}
RZ(\theta):=\begin{pmatrix}
e^{-i\theta/2} & 0 \\ 
0 & e^{i\theta/2} 
\end{pmatrix}.
\end{equation*}
It is easy to show that it cannot be obtained via a single physical $R$ pulse, and thus requires a circuit with two or more $R$ gates.  Firstly, we found the following circuit implementing the $RZ$ gate
\begin{eqnarray} \label{eq:rx-ry-rx}
RZ(\theta)= RY\left(-v\frac{\pi}{2}\right).RX\left(v\theta\right).RY\left(v\frac{\pi}{2}\right)  \\ \nonumber
= R\left(-v\frac{\pi}{2},\frac{\pi}{2}\right).R\left(v\theta, 0\right).R\left(v\frac{\pi}{2},\frac{\pi}{2}\right) ,
\end{eqnarray}
where $v \in \{-1, +1\}$ is a variable allowing to arbitrarily set the sign of either first or last $RY$ rotation, and `$.$' denotes matrix multiplication (recall that the order of gates in the circuit is given by the inverted order of matrices in the matrix product).  The implementation with $v=1$ was known to \cite{ar:deb}; as we will show later, the ability to choose $v$, being our contribution to the above circuit, is very important in circuit optimization.  The cost of this implementation is $\left(\frac{|\theta|\tau_{1q}}{\pi}+\tau_{1q}, 2 \times \epsilon + 1 \times |\sin\theta|\epsilon \right)$.  Alternatively, $RZ(\theta)$ gate may be obtained as
\begin{equation}\label{eq:rz}
RZ(\theta) \equiv R(\pi,x).R(\pi,x-\theta/2),
\end{equation}
up to an undetectable global phase of $-1$ (equality up to a global phase is furthermore denoted by `$\equiv$'), where the parameter $x$ may be set arbitrarily.  The cost of this implementation is $(2\tau_{1q}, 2 \times |\sin\pi|\epsilon) = (2\tau_{1q}, 0)$.  Observe that with the slightly longer execution time this second realization is associated with a smaller error, which seems to be a preferred scenario in the physical experiments if the gate is to be implemented by itself (as opposed to as a part of a larger computation).  The flexibility in setting $x$ within the above implementation allows to optimize quantum circuits where $RZ$ is one of the gates used, as varying the value $x$ allows to obtain either the $RX(\pi)$ gate or the $RY(\pi)$ gate to be either first or last gate in (\ref{eq:rz}), and those may cancel out with other gates in the circuit.  Varying parameter $x$ furthermore allows optimizing classical control, as selecting a value of the $R$ gate parameter used to implement a previous single-qubit gate allows to keep the phase of the Raman beatnote used a constant.  This, however, is only a minor improvement to the classical control sequences.  Implementation (\ref{eq:rx-ry-rx}) may become more desirable for the purpose of circuit optimization, since it relies on the efficiently optimizable sequence of $RX$ and $RY$ gates.

The $RZ$ gate may be implemented directly without resorting to a circuit-level composition of pulses by individually addressing the qubits with laser beams that result in a qubit energy level shift through the Stark effect \cite{ar:lhn,arXiv:1604.08840}.  Physical-level $RZ$ gate gives an advantage over the physical-level $R$ gate when one desires to construct the $RZ(\theta)$, as the $RZ(\theta)$ requires two physical-level $R$ gates, and only one physical-level $RZ$ gate to be implemented.  However, when such $RZ(\theta)$ is an internal gate to the circuit (and most gates in interesting quantum computations are internal), it can be written as either $R(\pi,0).R(\pi,-\theta/2)$ or $R(\pi,\theta/2).R(\pi,0)$ (\ref{eq:rz}), where $R(\pi,0)=RX(\pi)$ can be commuted past the two-qubit $XX$ gate (discussed in Section \ref{sec:2q}) selectably to either left or right.  This results in the effective ability to implement an internal $RZ(\theta)$ gate with just a single physical-level $R$ gate.  We furthermore note that in our optimized implementations of those circuits we tried, whenever the goal is to have no more than two sequential internal $R$ gates apply to a given qubit in a sequence, this was always possible to accomplish.  While this is unlikely to scale to arbitrary quantum computations, it is perhaps true that in practical designs most frequently no more than one $R$ gate is required between a pair of two-qubit $XX$ gates acting on a given qubit.  This illustrates the expressive power of the $R$ gates when used in conjunction with the two-qubit $XX$ gates.  To conclude this discussion, we believe the ability to implement $RZ$ directly may not be in high demand, unless the properties of such a physical-level $RZ$ gate, including its duration and error, are superior to the $R$ gate. 

$RX(\pi)$, $RY(\pi)$, and $RZ(\pi)$ implement Pauli-X, Pauli-Y, and Pauli-Z gates up to an undetectable global phase of $-i$. $RX(\pi/2)$ implements the square-root-of-NOT gate $V:=\frac{1+i}{2}\begin{pmatrix} \phantom{-}1 & -i \\ -i & \phantom{-}1 \end{pmatrix}$ up to a global phase.  The rotation $RZ(\pi/2)$ implements the quantum Phase gate (commonly referred to as $P$ or $S$) $P:=\begin{pmatrix} 1 & 0 \\ 0 & i \end{pmatrix}$ up to a global phase. The quantum $\pi/8$ gate also known as the $T$ gate, $T:=\begin{pmatrix} 1 & 0 \\ 0 & \frac{1+i}{\sqrt{2}} \end{pmatrix}$, is obtained as $RZ(\pi/4)$, up to a global phase. 

Recall that $RX$, $RY$, and $RZ$ gates do not commute, but their parameters may be added, {\it i.e.}, $G(a)G(b)=G(a+b)$, when $G$ is either one of $RX$, $RY$, or $RZ$.  This is important for the circuit optimization technique discussed later.

\subsection{Other common single-qubit rotations}
A common single-qubit gate that may not be expressed as an axial rotation with a certain parameter is the Hadamard gate, $H:=\frac{1}{\sqrt{2}}\begin{pmatrix} 1 & \phantom{-}1 \\ 1 & -1 \end{pmatrix}$. It can be implemented up to a global phase as one of the following two circuits.
\begin{eqnarray}\label{eq:h}
H \equiv RY(-\pi/2).RX(\pi) = R(\pi/2,\pi/2).R(\pi, 0) , \\
H \equiv RX(-\pi).RY(\pi/2) = R(-\pi, 0).R(\pi/2,\pi/2) . \nonumber
\end{eqnarray}
The cost of each of the above circuits is $(1.5 \tau_{1q},\epsilon)$.  As such, the Hadamard gate is, roughly speaking, as expensive as the $RZ$ rotation, and more expensive than either $RX$ or $RY$.  The ability to choose which of the $RX/RY$ gates in the decomposition of the Hadamard gate comes first and which is second is important to the optimization of quantum circuits.

\subsection{Arbitrary single-qubit rotations}
An arbitrary single-qubit unitary gate can be written as a matrix
$U = e^{id} \left( \begin{array}{rr}
e^{ia}\cos{b} & e^{ic}\sin{b}\\
-e^{-ic}\sin{b} & e^{-ia}\cos{b}
\end{array} \right)$
of four real-valued parameters $a,b,c,$ and $d$. We found the following implementation as a circuit with at most two physical-level $R$ gates, 
\begin{equation}\label{eq:any-single-qubit-gate}
U \equiv \left( \begin{array}{rr}
e^{ia}\cos{b} & e^{ic}\sin{b}\\
-e^{-ic}\sin{b} & e^{-ia}\cos{b}
\end{array} \right) = R(-\pi,-c-\pi/2).R(2b+\pi,a-c-\pi/2).
\end{equation}
The cost of this implementation is $(\tau_{1q} + \frac{|2b \bmod \pi|\tau_{1q}}{\pi}, |\sin(2b)|\epsilon )$.  Observe, that equation (\ref{eq:any-single-qubit-gate}) uses the minimal number of physical-level gates $R$ required to implement an arbitrary single-qubit unitary.  This can be established by counting the number of the real-valued degrees of freedom in the $2 \times 2$ unitary matrices up to global phase, and comparing it to the number of the real-valued degrees of freedom of the $R$ gates.  Equation (\ref{eq:any-single-qubit-gate}) furthermore gives rise to the following Lemma providing a guarantee on the cost of quantum physical-level circuits.

\vspace{1mm}\begin{lemma}\label{lem:1}
Any quantum physical-level circuit over $n$ qubits with $G$ two-qubit $XX$ gates can be reduced to an equivalent one with no more than $2(n+2G)$ single-qubit $R$ gates, providing, across all single-qubit gates used, an overall contribution of no more than $2\tau_{1q}(n+2G)$ to the runtime and a term of no more than $(n+2G) \times \epsilon$ to the error. 
\end{lemma}
\begin{proof}
First, count the number of the ``pieces of wire'', defined as an uninterrupted by any two-qubit gate qubit evolution time piece between two two-qubit gates in the circuit.  This number is given by the expression $n+2G$, as is easy to verify by induction on $G$, the number of the two-qubit gates. Indeed, there are $n$ pieces of wire in quantum circuits with no two-qubit gates, and the introduction of a two-qubit gate at the end of the circuit increases the number of the pieces of wire by two (specifically, on those qubits that the given two-qubit gate operates on). The single-qubit operations are contained to the individual pieces of wire, allowing to conclude the proof by referring to the equation (\ref{eq:any-single-qubit-gate}) and the definitions of the duration and error.
\end{proof}

The above Lemma reports an upper bound on the number of single-qubit gates, and could be used as a bottom line comparison for evaluating the efficiency of the optimizing compiler developed and tested in Sections \ref{sec:compiler}-\ref{sec:benchmarks}.  We furthermore observe that the above Lemma applies to show that the $5$-qubit $80$-gate QFT5 circuit, of which $10$ are two-qubit gates reported in \cite{ar:deb} is suboptimal.  This is because according to Lemma \ref{lem:1} the upper bound on the number of gates in such a circuit is $60$.

\section{Two-Qubit Gates}\label{sec:2q}
\subsection{Physical-level two-qubit gate}
The physical-level two-qubit gate available to us is the so-called $XX(\chi)$ gate, with parameter $\chi$ that depends on the pair of ions the gate is being applied to.  The gate itself is defined by the following unitary matrix \cite{ar:deb}:  
\begin{equation*}
XX(\chi)= \left( \begin{array}{rrrr}
\cos(\chi) & 0 & 0 & \hspace{-2mm}-i \sin(\chi)\\
0 & \hspace{-1mm}\cos(\chi) & \hspace{-3mm}-i \sin(\chi) & 0\\
0 & \hspace{-3mm}-i \sin(\chi) & \hspace{-1mm}\cos(\chi) & 0\\
-i \sin(\chi) & 0 & 0 & \cos(\chi)\\
\end{array} \right).
\end{equation*}

The absolute value of the phase, $|\chi|$, can be set to an arbitrary real number between $0$ and $\pi/2$ by varying the laser power used in the experiment \cite{ar:deb}.  The sign of $\chi$ depends on the laser detuning which is chosen based on the normal modes a particular pair of ions interacts with most strongly and hence which qubits the gate is being applied to \cite{ar:deb}.  The sign for each two-qubit gate is thus fixed experimentally and becomes an input parameter for how one is allowed to construct circuits. 

The $XX(\chi)$ gate implements the well-known M\o lmer-S\o rensen gate \cite{ar:somo}, and latter is known to generate the CNOT gate (for $\chi = \pm \pi/4$) using single-qubit operations on the input and the output side.  The CNOT gate is an important computational primitive---the ability to obtain it, coupled with the ability to implement any single-qubit gate, see equation (\ref{eq:any-single-qubit-gate}), gives computational universality \cite{ar:bbcd}.  However, the ability to vary parameter $\chi$ to accept values beyond $\pm \pi/4$ allows a more efficient implementation of important quantum gates.  Once computational universality is obtained, the efficiency becomes a next important step. 

The following property of the $XX(\chi)$ gate is important to note for future discussions: $XX(\chi)$ commutes with any single-qubit $RX(\theta)$ rotation.  However, $XX(\chi)$ does not commute with either $RY(\theta)$ or $RZ(\theta)$ when $\theta \not\equiv 0 \bmod 2\pi$.

\subsubsection{Cost function variables}

The two-qubit $XX(\chi)$ gate has the vector-function cost $(d,e)$, where \cite{pers}:
\begin{eqnarray}\label{eqn:2ed}
d:=\tau_{2q}, \text{ and } 
e=e_1:=|\sin(2\chi)|E \text{ or } e=e_2:=E.
\end{eqnarray}
Here, $d$ is the duration of the two-qubit $XX$ interaction, and $e$ either models the error due to fluctuations in the experiment, analogous to the single-qubit case ($e_1$), or simply accepts a constant value $E$ that is independent of the rotation angle ($e_2$).  Current experiment \cite{ar:deb} is setup such that the two-qubit $XX$ gates can be applied to any pair of qubits in a serial fashion. Compared to the single-qubit gate $R$, the two-qubit gate is substantially longer in runtime and has a higher error.  As a result, efficient circuit implementations must prefer the minimization of the use of the two-qubit $XX$ gate over minimizing the single-qubit gates. 

\subsection{Constructing the CNOT gate}

\begin{figure*}[t]
\centerline{
\begin{tabular}{c}
\Qcircuit @C=0.1em @R=1.5em @!R {
& \qw 	& \ctrl{1}	& \qw \\
& \qw		& \targ		& \qw
}
\raisebox{-1.4em}{\hspace{1mm}$\equiv$\hspace{1mm}}
\Qcircuit @C=0.3em @R=0.6em @!R {
& \qw 	& \gate{RY(v\frac{\pi}{2})}	& \multigate{1}{XX(s\frac{\pi}{4})} 	& \gate{RX(-s\frac{\pi}{2})} 	& \gate{RY(-v\frac{\pi}{2})}	& \qw \\
& \qw	& \qw						& \ghost{XX(s\frac{\pi}{4})} 			& \gate{RX(-vs\frac{\pi}{2})}	& \qw								& \qw
} \\ 
\end{tabular}
}
\caption{Implementation of the CNOT gate using physical-level gates, where $s = \pm 1$ is the sign of the interaction parameter $\chi$, specified by the ions the gate applies to ($s$ is a parameter that cannot be varied), and $v = \pm 1$ may be chosen arbitrarily.}\label{circ:cnot}
\end{figure*}
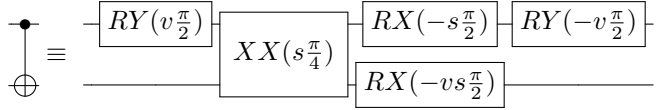

Depending on the sign $s$ of the interaction $\chi$ for the given pair of qubits that we want to apply the CNOT gate to, it can be implemented up to the global phase of $(-1)^{-\frac{vs}{4}}$ such as shown in Figure \ref{circ:cnot}.  Observe that $v$ may be chosen arbitrarily from the set $\{-1,+1\}$ to set the signs of the rotation angle in $RY$ at the beginning or at the end of the implementation.  In particular, parameter $v$ may be chosen such as to set the angle of the first $RY$ rotation to the positive number, $+\frac{\pi}{2}$, in which case the second $RY$ features the negative sign, or vice versa.  The ability to choose the sign is particularly important in allowing single-qubit $RY$ gate cancellations while decomposing logical-level circuits with multiple CNOT gates into efficient physical-level circuits. 

The cost of the above implementations of the CNOT gate is $(2\tau_{1q} + \tau_{2q}, 4 \times \epsilon + 1 \times E)$. 

Other than using two fewer single-qubit gate pulses compared to \cite{arXiv:1508.03392}, our CNOT implementation allows to arbitrarily set the values of the $RY$ rotations, and also allows further transformations and optimizations per template (\ref{eqn:single-qubit}) discussed later.  Our CNOT implementation saves two single-qubit pulses (one $RX(\pm\frac{\pi}{2})$ and one $RY(\pm\frac{\pi}{2})$) over the one reported in \cite{ar:deb}. 


\subsection{Constructing controlled-roots of Paulis}

Controlled roots of axial rotations (Pauli gates) play an important role in quantum circuits.  For instance, the $n$-qubit quantum Fourier transform is best viewed as a circuit with $\frac{n(n-1)}{2}$ controlled roots of Pauli-$Z$ gates \cite[Figure 5.1]{bk:nc}.  The controlled-sqrt-NOT gate is used in the construction of an efficient five two-qubit gate circuit implementing the Toffoli gate \cite[Lemma 6.1]{ar:bbcd}.  Otherwise, if the proper root is not available, and the CNOT gate is the only two-qubit gate directly constructible, the Toffoli gate requires six two-qubit physical gates \cite{ar:sm}.  It is furthermore known that each controlled unitary gate can be implemented with the use of two CNOT gates, and, equivalently, two M\o lmer-S\o rensen gates, along with some single-qubit gates \cite[Lemma 5.1]{ar:bbcd}.  We next show that only one $XX(\chi)$ gate suffices to implement any controlled root of a Pauli gate, providing an improvement by a factor of two.  In particular, depending on the sign $s$ of $\chi$, the controlled-$X^\alpha$, $\alpha \in \mathbb{R}, -1 \leq \alpha \leq 1,$ may be obtained as follows:
\begin{eqnarray}\label{circ:controlled-root}
\Qcircuit @C=0.2em @R=1em @!R {
& \qw 	& \ctrl{1}			& \qw \\
& \qw	& \gate{X^\alpha}	& \qw
}
\raisebox{-1.1em}{\hspace{1mm}$\equiv$\hspace{1mm}}
\Qcircuit @C=0.4em @R=0.6em @!R {
& \qw 	& \gate{RY(-s\frac{\pi}{2})}	& \multigate{1}{XX(s\frac{\alpha\pi}{4})} 	& \gate{RX(-s\frac{\alpha\pi}{2})} 	& \gate{RY(s\frac{\pi}{2})}	& \qw \\
& \qw	& \qw							& \ghost{XX(s\frac{\alpha\pi}{4})} 			& \gate{RX(\frac{\alpha\pi}{2})}	& \qw						& \qw
} 
\end{eqnarray}

The inverse of the controlled-$X^\alpha$ may be obtained from the above circuit by attempting to construct the controlled-$X^{-\alpha}$.  Observe that the decomposition with the sign opposite to the physically available sign of $\chi$ needs to be selected.  Other controlled roots of Paulis, such as the controlled-$Y^\alpha$ and the controlled-$Z^\alpha$ are related to the controlled roots of the NOT gate by the following formulas, 
\begin{eqnarray}\label{circ:cya-cza}
\Qcircuit @C=0.5em @R=1em {
& \ctrl{1} 				& \qw \\
& \gate{Y^\alpha} & \qw
}
&
\raisebox{-1em}{=}
&
\Qcircuit @C=0.7em @R=1em {
& \qw		& \ctrl{1} 				& \qw					& \qw \\
& \gate{P^{\dagger}} & \gate{X^\alpha} & \gate{P} & \qw
} \\ \nonumber
\Qcircuit @C=0.5em @R=1em {
& \ctrl{1} 				& \qw \\
& \gate{Z^\alpha} & \qw
}
&
\raisebox{-1em}{=}
&
\Qcircuit @C=0.7em @R=1em {
& \qw		& \ctrl{1} 				& \qw					& \qw \\
& \gate{H} & \gate{X^\alpha} & \gate{H} & \qw
} 
\end{eqnarray}

As a result, all controlled roots of Paulis are constructible using at most one physical-level two-qubit $XX$ gate. 

In our constructions, we favoured the decompositions that feature $RY(\pm\frac{\pi}{2})$ gates with arbitrarily selectable sign of the rotation parameter.  This was done to allow the selection of the specific sign or the relations between signs to dictate the $RY$ gate cancellations.  In some cases, the results of such cancellations can be dramatic, as is illustrated in the following Lemma.
\vspace{1mm}\begin{lemma}\label{lem:2}
Circuits over arbitrary $RZ$ rotations (including Pauli-$Z$, Phase, and $T$ gates) and the controlled-$Z$ can be written as an efficient trapped ions physical-level implementation as follows:
\begin{enumerate}
\item The layer of $RY$ gates, defined as the set of gates $RY(v_i\frac{\pi}{2})$ applied to every qubit $i$, where $v_i \in \{-1,+1\}$ can be selected arbitrarily.
\item The layer of $RX$, where qubit $i$ experiences the application of $RX(t-(v_{x_1}s_{ix_1}+v_{x_2}s_{ix_2}+...+v_{x_k}s_{ix_k})\frac{\pi}{2})$, where $t$ is the aggregate angle accomplished by the combination of $RZ$ rotations applied to this qubit, $k$ is the number of $CZ$ gates that apply to the qubit $i$, and $s_{ix1},s_{ix_2},...,s_{ix_k}$ are the signs of the respective interactions. 
\item The set of $XX$ gates, where each controlled-Z gate $CZ(a,b)$ in the original circuit is represented by $XX_{ab}(s_{ab}\frac{\pi}{4})$.
\item The layer of $RY$ gates, with $RY(-v_i\frac{\pi}{2})$ applied to the qubit $i$.
\end{enumerate}
\end{lemma} 
\begin{proof}
Construct the controlled-Z gate as follows:
\begin{eqnarray}\label{circ:cz}
\Qcircuit @C=0.2em @R=1em @!R {
& \qw 	& \ctrl{1}	& \qw \\
& \qw	& \gate{Z}	& \qw
}
\raisebox{-1.1em}{\hspace{1mm}$\equiv$\hspace{1mm}}
\Qcircuit @C=0.4em @R=0.6em @!R {
& \qw 	& \gate{RY(v_1\frac{\pi}{2})}	& \multigate{1}{XX(s\frac{\pi}{4})} 	& \gate{RX(-v_2s\frac{\pi}{2})} 	& \gate{RY(-v_1\frac{\pi}{2})}	& \qw \\
& \qw	& \gate{RY(v_2\frac{\pi}{2})}	& \ghost{XX(s\frac{\pi}{4})} 			& \gate{RX(-v_1s\frac{\pi}{2})}		& \gate{RY(-v_2\frac{\pi}{2})}	& \qw
}
\end{eqnarray} 
Observe that per formulas (\ref{eq:rx-ry-rx},\ref{circ:cz}) and upon the decomposition of $RZ$ and controlled-$Z$ gates into $RX/RY/XX$ circuits each qubit in each $RZ/CZ$ gate experiences the application of $RY(v\frac{\pi}{2})$ in the beginning and $RY(-v\frac{\pi}{2})$ in the end.  Therefore, setting the value of the parameter $v$ equal to the previously used $v$ accomplishes the following: every two $RY$ gates between any two $RZ/CZ$ gates in the target circuit cancel out.  However, there will be a layer of $RY$ in the beginning of the circuit and a layer of $RY$ at the end with the opposite signs, for each qubit that experiences an application of a gate.  Next, notice that all internal gates (those except the outside $RY$ layers) are $RX$ and $XX$, and thereby they all commute.  This means we need only one layer of $RX$, and due to the formulas (\ref{eq:rx-ry-rx},\ref{circ:cz}), the rotation angle is calculated such as stated in the Lemma.  Each $CZ(a,b)$ in the original circuit is furthermore represented by $XX_{ab}(s_{ab}\frac{\pi}{4})$, per formula (\ref{circ:cz}).
\end{proof}

Observe that allowing the layers of $RY$ gates in Lemma \ref{lem:2} to share one sign across all $RY$ used ($v_i=v_j$ for every pair of qubits $i$ and $j$) allows their implementation with a single global pulse. This capability is currently not supported by the existing hardware \cite{ar:deb}, but may potentially be implemented in the future.

\subsection{Useful single-qubit circuit identity}

\begin{figure}[t]
\centering
\begin{tabular}{cc}
\includegraphics[width=0.37\columnwidth]{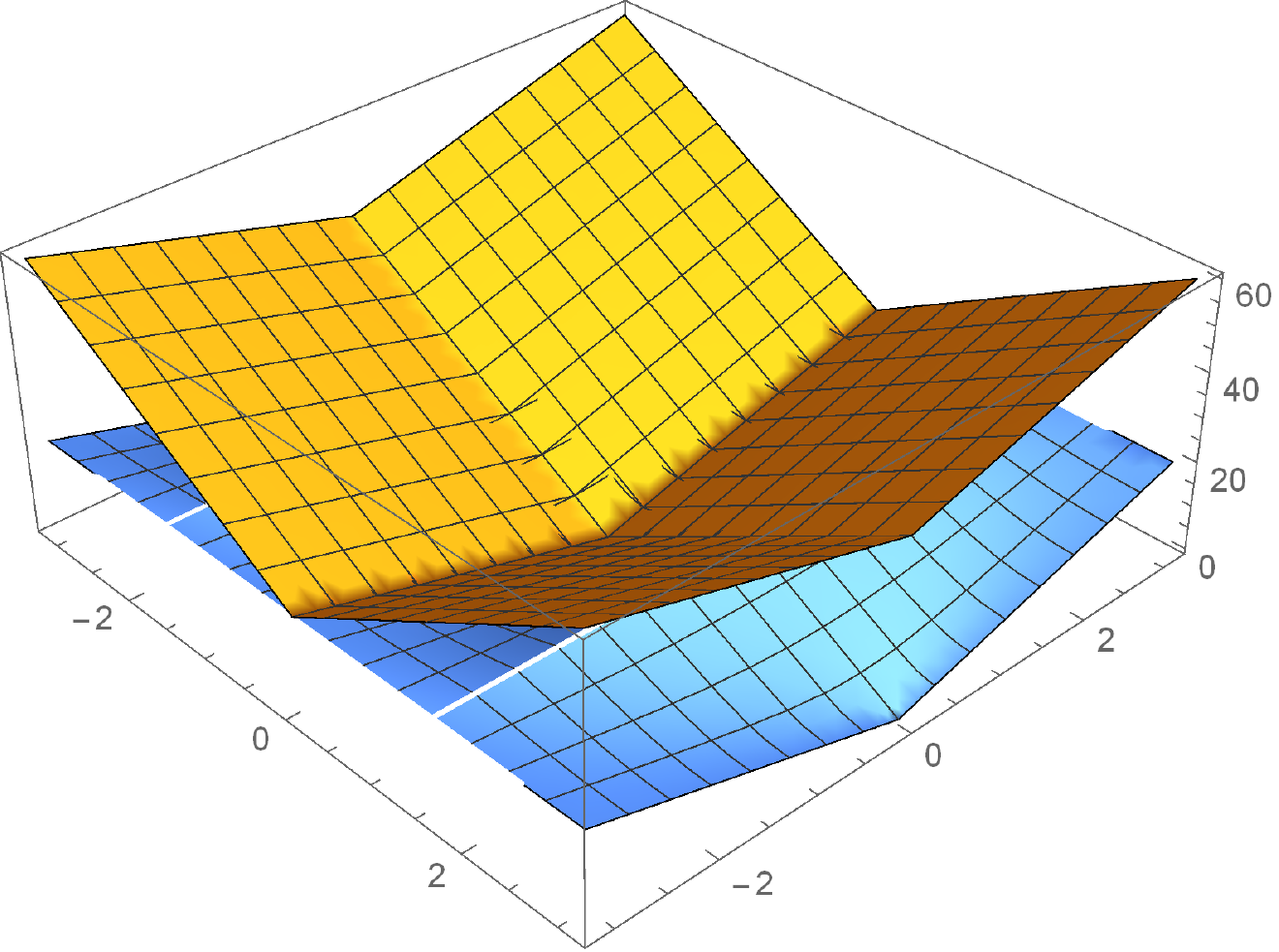} & \includegraphics[width=0.37\columnwidth]{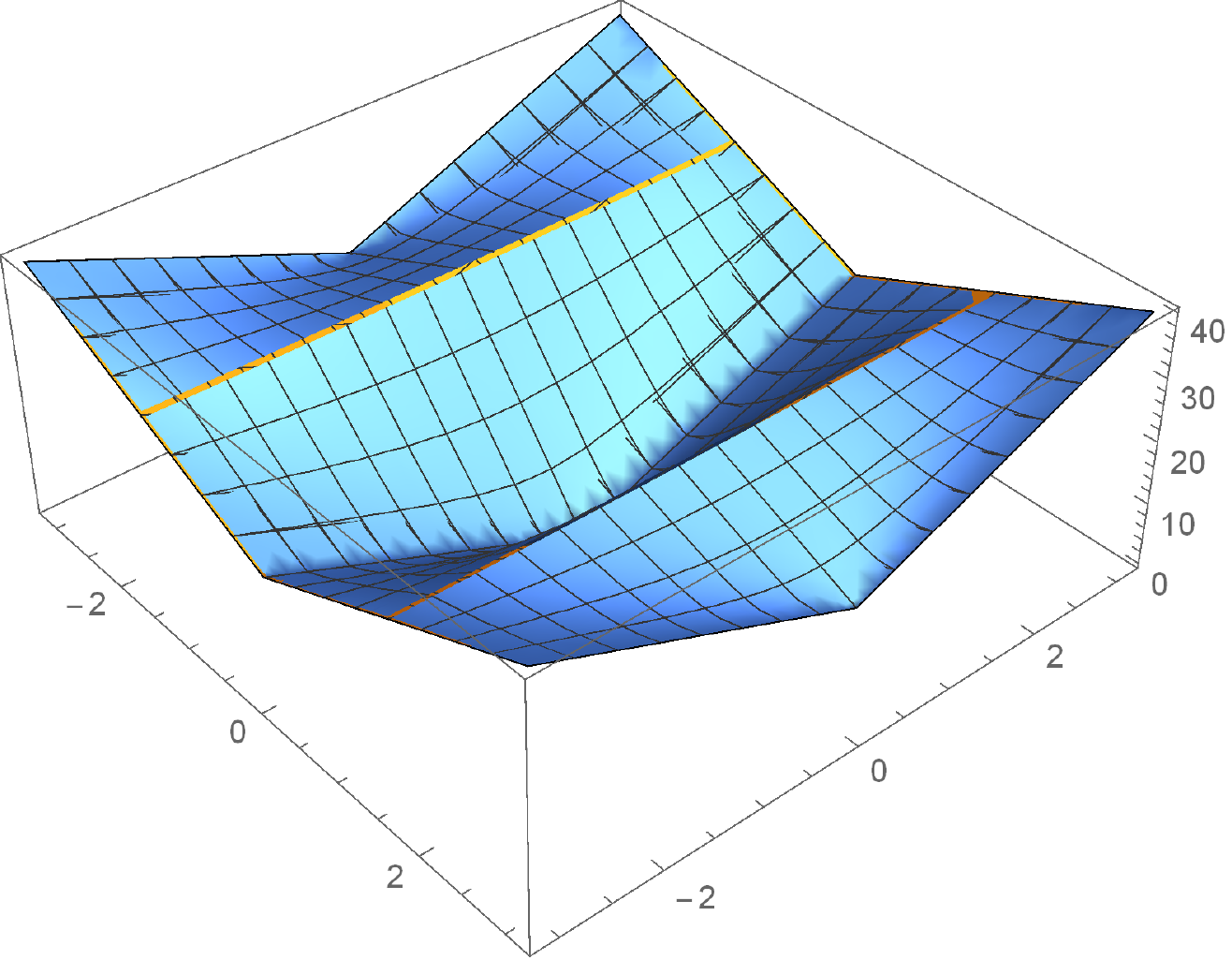} \\
(a) & (b) \\
\includegraphics[width=0.4\columnwidth]{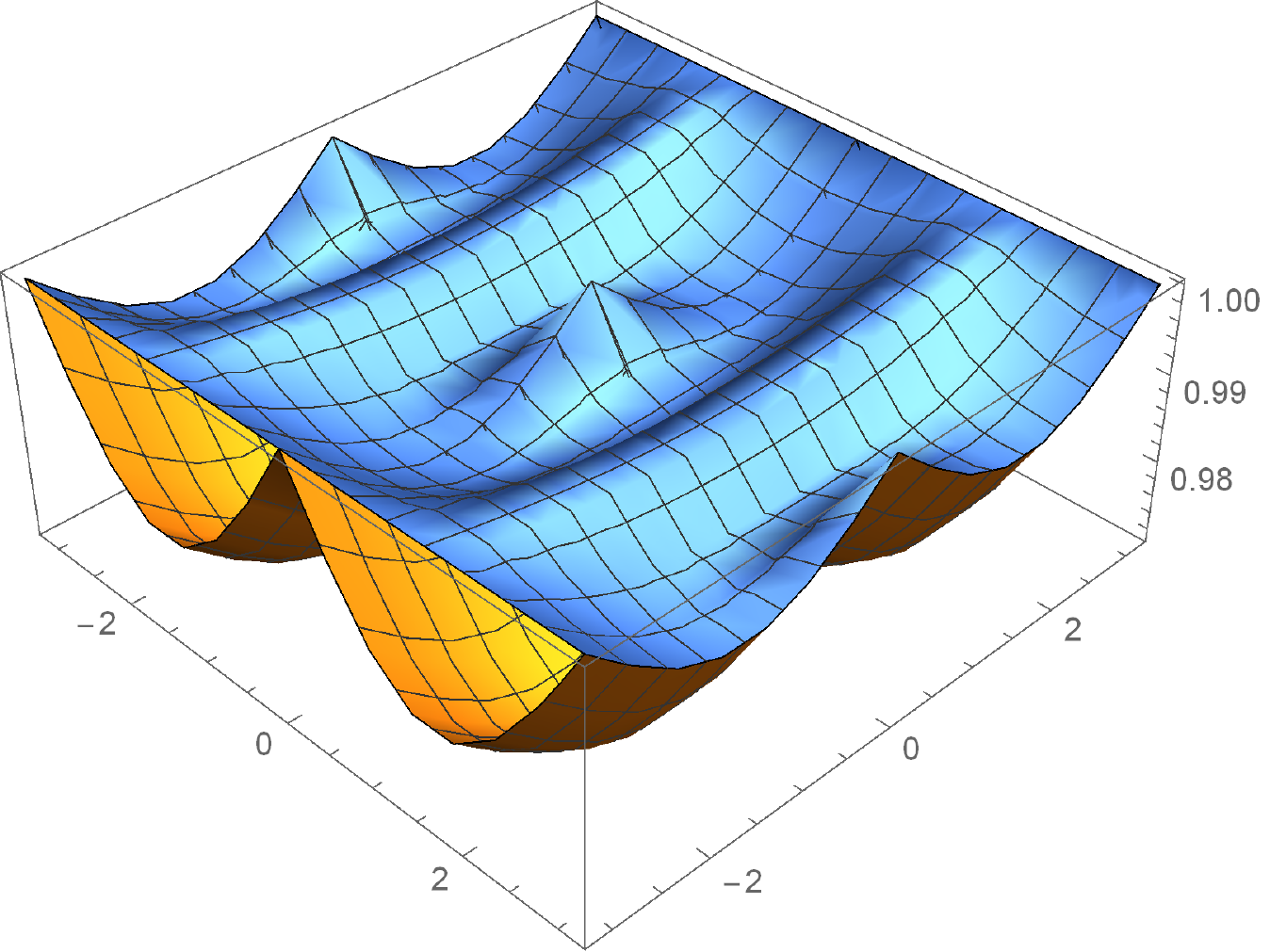} & \includegraphics[width=0.4\columnwidth]{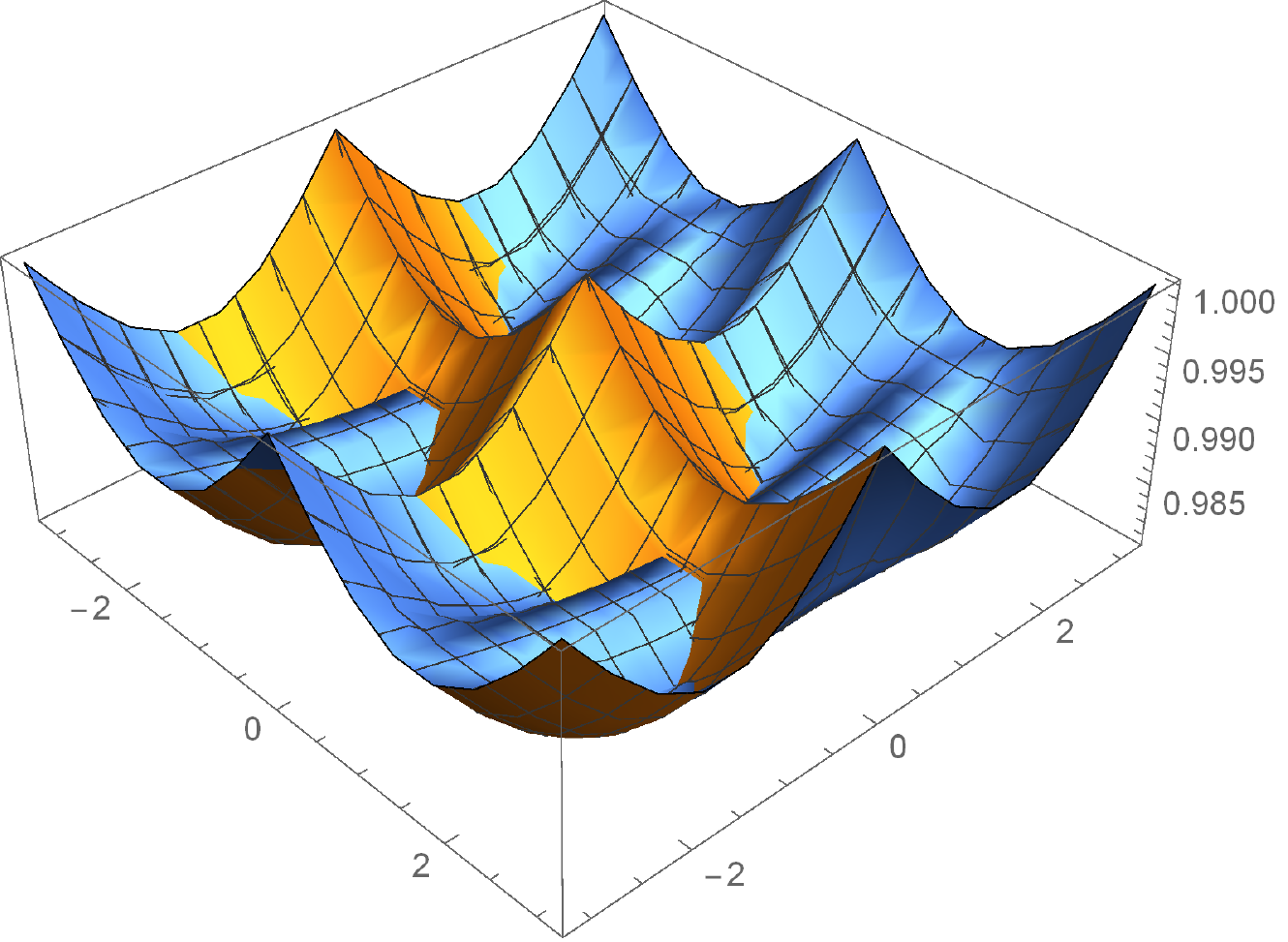} \\
(c) & (d) 
\end{tabular}
\caption{(a) change in the duration (old: yellow, new:blue; lower is better) in replacing $RX(a)RY(b)RX(a)$, $a,b \in [-\pi,\pi]$ by $R(c,d)$; 
(b) change in the duration (old: yellow, new:blue; lower is better) in replacing $RX(a)RY(b)$, $a,b \in [-\pi,\pi]$ by $R(c,d)RX(-a)$;
(c) change in the fidelity (old: yellow, new: blue; higher is better) in replacing $RX(a)RY(b)RX(a)$, $a,b \in [-\pi,\pi]$ by $R(c,d)$;
(d) change in the fidelity (old: yellow, new: blue; higher is better) in replacing $RX(a)RY(b)$, $a,b \in [-\pi,\pi]$ by $R(c,d)RX(-a)$.
For the purpose of this illustration, $\tau_{1q}$ was selected to be equal to $20 \mu s$ and $\epsilon = 0.01$, roughly corresponding to those values seen in a specific experiment \cite{ar:deb}.
Fidelity of the circuit spanning a single qubit with the gates $G_1,G_2,...,G_k$ featuring the individual errors $e_1,e_2,...,e_k$ is calculated as the product $\prod_{i=1}^k(1-e_i)$. This corresponds to the model where the errors are independent, and is consistent with the understanding of the physics and source of errors in the $R$ gates \cite{ar:deb, pers}.}
\label{fig:tem}
\end{figure}

Recall that a quantum template is a quantum circuit with $n$ gates that evaluates to the identity, $G_0G_1...G_{n-1} = Id$ \cite{ar:mdmn}.  A template can be used to construct a number of circuit identities, 
\begin{equation*}
G_iG_{i+1 \bmod n}...G_{i+k-1 \bmod n} = G^{-1}_{i-1 \bmod n} G^{-1}_{i-2 \bmod n} ... G^{-1}_{i+k \bmod n},
\end{equation*}
for arbitrary $i$ and $k$, $0 \leq i,k \leq n$, which may, in turn, be used to optimize quantum circuits via matching gates on the left hand side in the above equation and replacing them with the gates on the right hand side.  Here we report one such template that is particularly useful in our constructions.  Specifically, for any $a,b \in \mathbb{R}$ there exist $c,d \in \mathbb{R}$ such that:
\begin{equation}\label{eqn:single-qubit}
\Qcircuit @C=0.5em @R=1em {
& \gate{RX(a)}	& \gate{RY(b)}	& \gate{RX(a)} & \gate{R(c,d-\pi)} & \qw 
}
\hspace{1mm}\raisebox{-0em}{=}\hspace{2mm}
\Qcircuit @C=1em @R=1em {
	& \qw
}
\end{equation}
\begin{eqnarray*}
\text{where}\hspace{30mm} c := 2\arccos\left(\cos{a}\cdot\cos{\frac{b}{2}}\right), \\
d := \left\{ 
	\begin{array}{ll}
	\arcsin\left(\frac{\sin{\frac{b}{2}}}{\sqrt{1-\cos^2a\cdot\cos^2\frac{b}{2}}}\right) & \textrm{ if $a>0$} \\
	\pi-\arcsin\left(\frac{\sin{\frac{b}{2}}}{\sqrt{1-\cos^2a\cdot\cos^2\frac{b}{2}}}\right) & \textrm{ if $a<0$}. 
	\end{array} \right.
\end{eqnarray*} 
The above template may be used in multiple ways.
\begin{itemize}
\item Firstly, it allows the replacement of $RX(a)RY(b)RX(a)$ with $R(c,d)$.  The latter circuit always has a smaller duration and a higher overall fidelity (smaller contribution to the overall error), see Figure \ref{fig:tem}(a)(c).  
\item This template may be used to replace $RX(a)RY(b)$ with $R(c,d)RX(-a)$ and $RY(b)RX(a)$ with $RX(-a)$ $R(c,d)$.  This allows to trade off runtime for error, see Figure \ref{fig:tem}(b)(d) for the illustration of the changes in the runtime and error. While the runtime always increases, it is sometimes possible to improve the fidelity.  For instance, if $RX(\frac{\pi}{2})RY(-\frac{\pi}{2})$ in the circuit in Figure \ref{circ:cnot} were replaced with $R(\pi,-\frac{\pi}{4})RX(-\frac{\pi}{2})$ this would result in the cost vector change of that part of the computation from $(\tau_{1q},2 \times \epsilon)$ to $(1.5\tau_{1q}, \epsilon)$.  This constitutes an increase of the runtime by $0.5 \tau_{1q}$ over the increase of the fidelity from $(1-\epsilon)^2$ to $1- \epsilon$.  Since the above rule applies to any pair $RX$ and $RY$, it allows substantial flexibility in exchanging runtime for error.
\end{itemize}

In our approach to the optimization of quantum circuits we favour gate decompositions relying on $RX$, $RY$, and $XX$ gates.  The efficiency is evidenced through short elementary gate decompositions ($Z$ and its roots, Hadamard, CNOT, controlled-roots of Paulis), and favourable in-circuit gate cancellations, such as illustrated in Lemma \ref{lem:2}.  Template (\ref{eqn:single-qubit}) enables the next set of optimizations.  Specifically, given a circuit over $RX$, $RY$, and $XX$ gates, the template (\ref{eqn:single-qubit}) allows to ``commute'' arbitrary $RX$ to the right (or left) of every $RY$ met via replacing $RX(a)RY(b)$ with the properly defined $R(c,d)RX(-a)$. Observe that such ``commutation'' changes the sign of the parameter in $RX$, as such it is a special kind of commutation.  Recalling that $RX(a)RX(b) = RX(a+b \bmod 2\pi)$ and that $RX$ commutes with $XX$ allows to ``commute'' all $RX$ to the end (or beginning) of the circuit via replacing $RY$ gates with $R$ gates.  This reduces the number of $RX$ gates to at most one per qubit.  This result is furthermore summarized in the following Lemma. 

\vspace{1mm}\begin{lemma}\label{lem:3}
Any quantum physical-level circuit over $n$ qubits with $G_{XX}$ two-qubit $XX$ gates, $G_{RY}$ single-qubit $RY$ gates, and $G_{RX}$ single-qubit $RX$ gates can be reduced to an equivalent one with no more than $G_{XX}$ two-qubit $XX$ gates, $G_{RY}$ single-qubit $R$ gates, and at most $n$ single-qubit $RX$ gates. 
\end{lemma} 


\section{Compiling quantum algorithms into physical-level circuits}\label{sec:compiler}

Define two circuit cost metrics, a coarse-grain and a fine-grain one.  The coarse-grain metric counts the number of the two-qubit controlled roots of Paulis in quantum circuits.  The fine-grain metric is described by the cost vector $(time,error)$, where $time$ sums up the runtimes across all gates used, and $error$ combines all errors.  Note that once the controlling apparatus allowing to execute gates in parallel is developed, the definition of $time$ will change into the sum of times across the critical path; other changes to the above costing metrics may also be accommodated, and depend on the improvements in the controlling apparatus and/or adjustments made to the error model. 

The following reports all steps taken by our overall design approach that maps a quantum algorithm into an optimized physical-level experiment. 
\begin{enumerate}
\item Choose an algorithm and map it into a high-level logical circuit with the help of a quantum programming language, if needed. 
\item Synthesize all arithmetic and oracle parts, if not explicitly supplied.  Arithmetic circuits are chosen from the known libraries, and oracles are synthesized using known reversible logic synthesis algorithms \cite{ar:sama}.  Optimize the resulting implementation over the coarse-grain cost metric \cite{ar:mdm}. 
\item Decompose the multiple-control Toffoli gates into smaller gates, such as three-qubit Toffoli gate and small relative phase Toffoli gates \cite{ar:m}.  Optimize circuits using peep-hole \cite{ar:pspmh} and templates \cite{ar:mdmn} over the coarse-grain cost metric.
\item Break down all gates into two-qubit controlled roots of Paulis and arbitrary single-qubit gates using optimal implementations \cite{bk:nc}, and optimize the resulting decompositions using templates \cite{ar:mdmn} over the coarse-grain metric.  At the end of this stage we should have reached the limit of optimization of the number of most expensive two-qubit gates, therefore we will next switch the gear and employ the fine-grain metric. 
\item Map logical qubits into physical qubits such as to minimize the use of the least desired interactions (those associated with the least two-qubit $XX$ gate fidelities), and maximize gate cancellations during further optimization.  To accomplish latter, record the position and the signs of all $RY(\pm \frac{\pi}{2})$ gates participating in the expansions of the powers of the controlled roots of Paulis that cannot be varied but depend of the sign of $\chi$ (\ref{circ:controlled-root}), and favour the selection of physical qubit mapping resulting in $RY$ cancellations on the control qubit.  The cancellation happens between two controls when there are no gates between them, and the signs of $\chi$ corresponding to the two $XX$ rotations are equal.  In general, so long as the number of qubits remains small, this can be done exhaustively; otherwise, a mix of subgraph isomorphism and greedy ({\em i.e.}, those making a local choice at each stage in hopes this has little effect on the global optimality) heuristics can be employed. 
\item Decompose further into physical-level circuit and optimize the resulting implementation: 
\begin{enumerate}
\item Perform controlled root of Pauli gate substitutions---those are now uniquely defined per formulas (\ref{circ:controlled-root}).  Decompose all single-qubit gates except Hadamard and $RZ$ into circuits over $RX/RY$. Choose CNOT (Figure \ref{circ:cnot}), controlled-Y (\ref{circ:cya-cza}), controlled-Z (\ref{circ:cz}), Hadamard (\ref{eq:h}), and $RZ$ (\ref{eq:rx-ry-rx}) gate decompositions such as to maximize the cancellations of pairs $RY(\pm \frac{\pi}{2})$ and $RY(\mp \frac{\pi}{2})$.  Perform cancellation of the $RY$ gates and combine the parameters of the neighbouring $RX/RY/XX$.
\item Apply template (\ref{eqn:single-qubit}) to reduce the selected figure of merit: to optimize gate count and error, ``commute'' $RX$ to one side of the circuit (incidentally, while error model was chosen to ensure a conflict of optimization across runtime and error, error optimization results in the reduction of the runtime; this is perhaps not surprising as $RX$ ``commutation'' has the effect of significantly reducing the $RX$ gate count, and both errors and runtime, as modelled, originate from the application of gates); to optimize the duration, find $RY(b)$ gates such that $RX$ with equal or similar parameters can be commuted to it from both left and right, and replace $RX(a)RY(b)RX(a)$ with proper $R(c,d)$.
\item Perform further balancing of runtime vs error using the template (\ref{eqn:single-qubit}) until the desired balance is found or no more improvement can be achieved.  
\item If the single-qubit gate sequences of length $3$ and higher are found, replace them with $2-R$ gate sequences, per formula (\ref{eq:any-single-qubit-gate}).  Rewrite all remaining single-qubit $RX$ and $RY$ gates as the physical-level $R$ pulses. 
\end{enumerate}
\end{enumerate}
The bulk of work and the optimization not previously explicitly considered in the literature falls into steps 5 and 6 of the above approach.  The complexity of the algorithms employed in step 6 is described by a low degree polynomial in the number of gates in the circuit (the degree of the polynomial is one in the simplest case of greedy algorithms).  The complexity of the algorithms used in step 5 depends on the efficiency of heuristics employed.  Steps 1-4 rely on a combination and a modification of the known techniques.

\section{Benchmark results}\label{sec:benchmarks}

The physical trapped ions machine we have access to \cite{ar:deb} currently has the following parameters: $\tau_{1q} = 20 \mu s$ and $\tau_{2q} = 235 \mu s$.  The single-qubit and two-qubit gate errors are approximately $\epsilon\sim 0.01$ and $E \sim 0.04$.  The value of $E$ may furthermore vary slightly depending on the set of qubits the gate is being applied to.  The signs of $\chi$, specifically, $\chi_{i,j}$, depending on the particular interaction between qubits $i$ and $j$ used, are as follows: $\chi$ has a positive sign for the interactions (next showing pairs of qubit numbers from the set of five, $\{1,2,3,4,5\}$) 12, 14, 23, 25, 34, 35, and 45.  The sign is negative for the interactions 13, 15, and 24.  Due to the values of physical errors, we expect to be able to apply circuits containing no more than about 15 two-qubit gates, therefore the bulk of work in this section is devoted to developing and optimizing experiments that satisfy the above condition.  In particular, our goal in this section is to propose experiments maximally utilizing the capabilities of the {\em specific} trapped ions machine \cite{ar:deb}, as well as to design those experiments with maximal efficiency while following precisely those procedures described in the previous sections.  In the coming subsection we illustrate how the above circuit compilation techniques accomplish the task of implementing the Boolean multiplication, and then report the results of the design of advanced experiments.  Majority of the computations proposed here scale beyond those previously demonstrated in an experiment. 

A relevant work on the circuit optimization for trapped ions technology was performed in \cite{ar:nhr}, that relied on the use of gradient ascend type algorithm to optimize the gate sequences.  Our optimization is based on establishing the relation between gate decompositions, and applying the template (\ref{eqn:single-qubit}), and therefore we expect the computational complexity as well as the practical efficiency of our algorithms to be better compared to the gradient-ascend type techniques.  Specifically, the algorithm for single-qubit circuit optimization with the purpose of error or gate count reduction employed in our work is described by the linear function of the number of gates, and the algorithm for the two-qubit gate optimization is described by at most qubic term, and furthermore allows the reduction to a linear complexity with the minimal loss to the quality of the output \cite{ar:mdmn}.  While the physical-level gates used in our work appear to be different from the physical-level gates employed in \cite{ar:nhr}, and no direct comparison can be made, we note that our CNOT implementation (Figure \ref{circ:cnot}) contains $1$ $XX$ gate and $4$ single-qubit pulses, whereas \cite[Section IV.A]{ar:nhr} reports a 2-$XX$ and 8 single-qubit pulse implementation, and our physical-level implementation of the circuit CNOT$[1,2]$CNOT$[1,3]=RY[1](-\frac{\pi}{2}).RX[3](-\frac{\pi}{2}).XX[1,2](-\frac{\pi}{4}).XX[1,3](\frac{\pi}{4}).RY[1](\frac{\pi}{2}).RX[2](\frac{\pi}{2})$ contains 2 $XX$ gates and 4 single-qubit pulses, whereas \cite[Section IV.A]{ar:nhr} reports a 2-$XX$ and 7 single-qubit pulse implementation.

A recent paper \cite{ar:mmns} reports a numeric optimization approach to designing trapped ions circuits over global Molmer-Sorenson gates (defined as $XX_{i,j}(\chi)$ applied to every pair of ions $i$ and $j$, with controllable parameter $\chi$), global $RX$ and $RY$ gates, as well as local single-qubit $RZ$ rotations.  In contrast, our work focuses on quantum computations by local gates.  The two (local vs global) are very different types of control and the corresponding circuits are incomparable: for instance, observe that our circuits work independently of the number of qubits the computation runs over, whereas circuits in \cite{ar:mmns} in general depend on the number of qubits used.  This said, global control can be used in a way independent of the number of qubits involved in the computation.  Specifically, an arbitrary local entangling \textsc{CNOT} gate can be constructed with the use of at least two global Molmer-Sorenson interactions and some number of single-qubit gates, allowing to express all quantum algorithms (majority of which are, in fact, described in terms of local operations) using global entangling gates.  However, the number of global Molmer-Sorenson gates used would then be twice as much as the number of local Molmer-Sorenson gates required to accomplish the same task using local control.

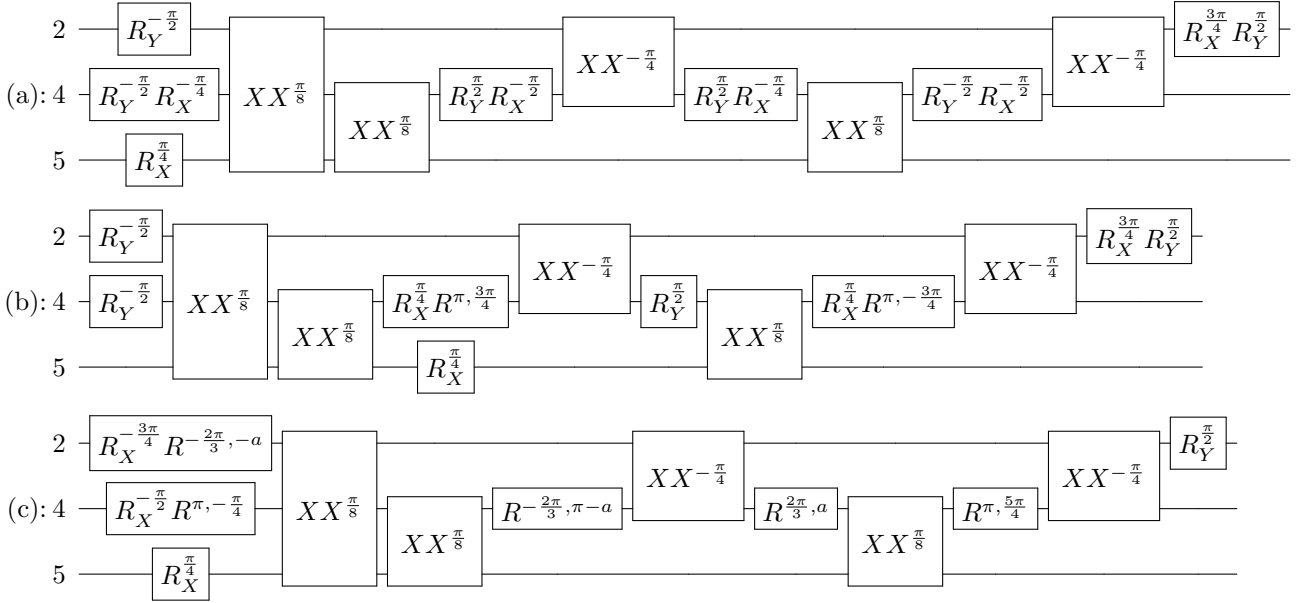
\begin{figure*}[t]
\centerline{
\begin{tabular}{rl}
\raisebox{-2.8em}{(a):} &
\Qcircuit @C=0.4em @R=0.4em @!R {
\lstick{2} & \gate{R_Y^{-\frac{\pi}{2}}}	& \multigate{2}{XX^{\frac{\pi}{8}}}    & \qw 								& \qw 															& \multigate{1}{XX^{-\frac{\pi}{4}}}  	& \qw & \qw 								& \qw											& \multigate{1}{XX^{-\frac{\pi}{4}}}	& \gate{R_X^{\frac{3\pi}{4}}R_Y^{\frac{\pi}{2}}} & \qw \\
\lstick{4} & \gate{R_Y^{-\frac{\pi}{2}}R_X^{-\frac{\pi}{4}}} & \ghost{XX^{\frac{\pi}{8}}} 			& \multigate{1}{XX^{\frac{\pi}{8}}}  & \gate{R_Y^{\frac{\pi}{2}}R_X^{-\frac{\pi}{2}}}  & \ghost{XX^{-\frac{\pi}{4}}}			& \gate{R_Y^{\frac{\pi}{2}}R_X^{-\frac{\pi}{4}}} & \multigate{1}{XX^{\frac{\pi}{8}}}	& \gate{R_Y^{-\frac{\pi}{2}}R_X^{-\frac{\pi}{2}}} 	& \ghost{XX^{-\frac{\pi}{4}}}		& \qw				& \qw \\
\lstick{5} &  \gate{R_X^{\frac{\pi}{4}}} 					& \ghost{XX^{\frac{\pi}{8}}} 			& \ghost{XX^{\frac{\pi}{8}}}			& \qw 															& \qw 		& \qw & \ghost{XX^{\frac{\pi}{8}}}			& \qw 		& \qw 			& \qw 			& \qw
} \vspace{3mm} \\ 
\raisebox{-2.8em}{(b):} &
\Qcircuit @C=0.4em @R=0.4em @!R {
\lstick{2} & \gate{R_Y^{-\frac{\pi}{2}}}	& \multigate{2}{XX^{\frac{\pi}{8}}}    & \qw 								& \qw 															& \multigate{1}{XX^{-\frac{\pi}{4}}}  	& \qw & \qw 								& \qw											& \multigate{1}{XX^{-\frac{\pi}{4}}}	& \gate{R_X^{\frac{3\pi}{4}}R_Y^{\frac{\pi}{2}}} & \qw \\
\lstick{4} & \gate{R_Y^{-\frac{\pi}{2}}} & \ghost{XX^{\frac{\pi}{8}}} 			& \multigate{1}{XX^{\frac{\pi}{8}}}  & \gate{R_X^{\frac{\pi}{4}}R^{\pi,\frac{3\pi}{4}}}  & \ghost{XX^{-\frac{\pi}{4}}}			& \gate{R_Y^{\frac{\pi}{2}}} & \multigate{1}{XX^{\frac{\pi}{8}}}	& \gate{R_X^{\frac{\pi}{4}}R^{\pi,-\frac{3\pi}{4}}} 	& \ghost{XX^{-\frac{\pi}{4}}}		& \qw				& \qw \\
\lstick{5} & \qw 					& \ghost{XX^{\frac{\pi}{8}}} 			& \ghost{XX^{\frac{\pi}{8}}}			&  \gate{R_X^{\frac{\pi}{4}}}		& \qw 		& \qw & \ghost{XX^{\frac{\pi}{8}}}			& \qw 		& \qw 			& \qw 			& \qw
} \vspace{3mm} \\
\raisebox{-2.8em}{(c):} &
\Qcircuit @C=0.4em @R=0.4em @!R {
\lstick{2} & \gate{R_X^{-\frac{3\pi}{4}}R^{-\frac{2\pi}{3},-a}}	& \multigate{2}{XX^{\frac{\pi}{8}}}    & \qw 								& \qw 															& \multigate{1}{XX^{-\frac{\pi}{4}}}  	& \qw & \qw 								& \qw											& \multigate{1}{XX^{-\frac{\pi}{4}}}	& \gate{R_Y^{\frac{\pi}{2}}} & \qw \\
\lstick{4} & \gate{R_X^{-\frac{\pi}{2}}R^{\pi,-\frac{\pi}{4}}} & \ghost{XX^{\frac{\pi}{8}}} 			& \multigate{1}{XX^{\frac{\pi}{8}}}  & \gate{R^{-\frac{2\pi}{3},\pi-a}}  & \ghost{XX^{-\frac{\pi}{4}}}			& \gate{R^{\frac{2\pi}{3},a}} & \multigate{1}{XX^{\frac{\pi}{8}}}	& \gate{R^{\pi,\frac{5\pi}{4}}} 	& \ghost{XX^{-\frac{\pi}{4}}}		& \qw				& \qw \\
\lstick{5} &  \gate{R_X^{\frac{\pi}{4}}} 					& \ghost{XX^{\frac{\pi}{8}}} 			& \ghost{XX^{\frac{\pi}{8}}}			& \qw 															& \qw 		& \qw & \ghost{XX^{\frac{\pi}{8}}}			& \qw 		& \qw 			& \qw 			& \qw
}
\end{tabular}
}
\caption{Toffoli gate on ions 2, 4 (controls), and 5 (target) implemented as a physical-level trapped ions circuit (rotation angles showing as superscripts) up to an undetectable global phase: (a) Toffoli gate implementation optimized for time alone, with gate parameters aggregated; (b) circuit in (a) optimized for time as primary parameter, and error as secondary parameter; (c) circuit (a) optimized for error as primary parameter, and time as secondary parameter, where the value $a:=\arcsin{\sqrt{\frac{2}{3}}} \approx 0.304087\pi$. The runtimes and errors ($e_1$) are: \newline
(a) $1285$ $\mu s$, $4\times 0.707107\epsilon +  \hspace{25mm} 8\times\epsilon + 3 \times 0.707107E + 2 \times E$;\newline
(b) $1285$ $\mu s$, $4\times 0.707107\epsilon +  \hspace{25mm} 4\times\epsilon + 3 \times 0.707107E + 2 \times E$;\newline
(c) $1295$ $\mu s$, $\hspace{0.3mm}2\times 0.707107\epsilon + 3\times 0.866025\epsilon + 2\times\epsilon + 3 \times 0.707107E + 2 \times E$.
} \label{circ:tofxxrxry}
\end{figure*}

\subsection{Implementing Boolean multiplication}

To implement the Boolean multiplication on the trapped ions machine our algorithm takes the following steps.

\noindent {\it 1-4.} Steps {\it 1-3} identify that the computation we wish to perform is given by the Toffoli gate, $\textsc{TOF}[a,b;c]$, and step {\it 4} finds the following 5 two-qubit gate circuit implementing the Toffoli gate \cite{ar:bbcd} using CNOT and controlled-$\sqrt{NOT}$ gates,
\[
\Qcircuit @C=0.6em @R=1.05em @!R {
& \ctrl{1} 	& \qw \\
& \ctrl{1} 	& \qw \\
& \targ 	& \qw
}
\hspace{1mm}\raisebox{-1.9em}{=}\hspace{1mm}
\Qcircuit @C=0.7em @R=0.4em @!R {
& \ctrl{2} 		& \qw 			& \ctrl{1}	& \qw				& \ctrl{1} 	&\qw \\
& \qw			& \ctrl{1} 		& \targ		& \ctrl{1}			& \targ		&\qw \\
& \gate{V} 		& \gate{V} 		& \qw 		& \gate{V^\dagger}	& \qw 		&\qw
}
\] 

\noindent {\it 5.} Physical qubits from the set $\{1,2,3,4,5\}$ are mapped onto logical qubits: we choose 2, 4, and 5 (top to bottom), such as to rely on the high fidelity interactions \cite[Table 1]{ar:deb}.  The controlled roots-of-NOT gates do not feature neighbouring controls, imposing no further conditions on the mapping. 

\noindent {\it 6.}
\begin{enumerate}
\vspace{-6.3mm}\item[a.] Now that the signs of $\chi_{i,j}$ are known, mark the controls of the controlled-$V/V^{\dagger}$ gates with $+/-$ by the sign of the $RY$ rotation that appears on it when decomposed into physical pulses using formulas (\ref{circ:controlled-root}),
\[
\Qcircuit @C=0.7em @R=0.4em @!R {
\lstick{2} & \ctrl{2}^{\hspace{3mm}\mbox{\tiny{-}\hspace{2mm}\tiny{+}}\hspace{-3mm}} 		& \qw 						& \ctrl{1}	& \qw				& \ctrl{1} 	& \qw \\
\lstick{4} & \qw			& \ctrl{1}^{\hspace{4mm}\mbox{\tiny{-}\hspace{2mm}\tiny{+}}\hspace{-4mm}}	& \targ		& \ctrl{1}^{\hspace{3.5mm}\mbox{\tiny{+}\hspace{2mm}-}\hspace{-3.5mm}}			& \targ		&\qw \\
\lstick{5} & \gate{V} 		& \gate{V} 					& \qw 		& \gate{V^\dagger}	& \qw 		& \qw
}
\] 
and then choose the CNOT gate decompositions (Figure \ref{circ:cnot}) such as to maximize $RY$ cancellations,
\[
\Qcircuit @C=0.7em @R=0.4em @!R {
\lstick{2} & \ctrl{2}^{\hspace{3mm}\mbox{\tiny{-}\hspace{2mm}\tiny{+}}\hspace{-3mm}} 		& \qw 						& \ctrl{1}^{\hspace{3.5mm}\mbox{\tiny{-}\hspace{2mm}\tiny{+}}\hspace{-3.5mm}}	& \qw				& \ctrl{1}^{\hspace{4mm}\mbox{\tiny{-}\hspace{2mm}\tiny{+}}\hspace{-4mm}} 	& \qw \\
\lstick{4} & \qw			& \ctrl{1}^{\hspace{4mm}\mbox{\tiny{-}\hspace{2mm}\tiny{+}}\hspace{-4mm}}	& \targ		& \ctrl{1}^{\hspace{3.5mm}\mbox{\tiny{+}\hspace{2mm}-}\hspace{-3.5mm}}			& \targ		&\qw \\
\lstick{5} & \gate{V} 		& \gate{V} 					& \qw 		& \gate{V^\dagger}	& \qw 		& \qw
}
\] 
The above choice for the CNOT decompositions allows to cancel two pairs of $RY(\frac{\pi}{2})/RY(-\frac{\pi}{2})$ gates on the first qubit, as is evidenced by the ``meeting'' plus and minus signs.
\item[b.] Perform gate substitutions and cancellations to obtain the circuit shown in Figure \ref{circ:tofxxrxry}(a).
\item[c.] The rest of the single-qubit optimization algorithm is based on the template (\ref{eqn:single-qubit}), and works differently depending on the criteria for the remaining optimization.  If the runtime is the goal, the algorithm tries to find triples of gates $RX-RY-RX$ that may be replaced with the $R$ gate, such as to minimize the overall duration.  This allows to construct the circuit pictured in Figure \ref{circ:tofxxrxry}(b).  In case if minimizing the error is preferred, the algorithm ``commutes'' all $RX$ to the left.  The result of this optimization may be found in Figure \ref{circ:tofxxrxry}(c).  We conclude the optimization by layering the single-qubit gates sharing the same duration, by moving $RX$, whenever possible, such as to allow their sequential execution---this helps to optimize classical controlling sequences. 
\end{enumerate} 

Observe that at the stage of the decomposition of the two-qubit logical gates into implementable physical-level gates, the single-qubit pulse count was $20$.  The parametrized CNOT implementation per Figure \ref{circ:cnot}, along with the algorithm for joint gate decomposition, and $RX$ gates commutation over the application of template (\ref{eqn:single-qubit}) allowed to reduce the single-qubit pulse count from the original $20$ down to $9$ (Figure \ref{circ:tofxxrxry}(c)).  Had the previously known circuitry implementing the CNOT gate with two more single-qubit pulses \cite{ar:deb, arXiv:1508.03392} been used, the original unoptimized single-qubit gate count would have been $30$.  The upper bound on the number of single-qubit gates in a circuit of this size, as given by the application of Lemma \ref{lem:1}, is $26$.  Had the efficient controlled-root-of-Pauli implementation introduced in this paper been not used, the resource count would have been substantially higher---not only the single-qubit gates, but this time, two-qubit gates as well \cite{ar:sm}.  This illustrates the power of the approach reported in this paper.

Our Toffoli gate implementation, Figure \ref{circ:tofxxrxry}(b), takes time $1285$ $\mu s$, which compares favourably to $1500$ $\mu s $ reported in \cite{ar:mkh}.  The advantage in the duration, however, is attributed to the different hardware \cite{ar:deb} that we rely on in our work.  The difference in the fidelity between our implementation and the one reported in \cite{ar:mkh} needs to be established via the experiment.  The major difference between our implementation and the one reported in \cite{ar:mkh} is our circuit is a true Toffoli gate that can be used (and in fact is used in the next subsection) as a primitive in the implementation of quantum algorithms, whereas \cite{ar:mkh} reports a Toffoli gate implemented up to a relative phase; such an implementation up to the relative phase may not be used in quantum algorithms directly---specifically, it would give an incorrect answer if used within Grover's search \cite{co:g}.

\subsection{Advanced experiments}

\begin{table*}[t]
\centering
{\footnotesize
\begin{tabular}{|l|c|c|c|r|}
\hline
\hspace{-1mm}Function 				& \hspace{-2mm}\#q\hspace{-2mm} & \hspace{-3mm}1q/2qg\hspace{-3mm} & Time 			& Error \\ \hline
\hspace{-1mm}Grover$^{\{011,111\}}$\hspace{-1.5mm} 	& 4 & 29/10	& \hspace{-1.5mm}2743$\mu s$\hspace{-1.5mm}	& $e_1 = 6 \times 0.707107\epsilon + 4 \times 0.866025\epsilon + 9 \times\epsilon + 6 \times 0.707107E + 4 \times E$ \\ 
&&&& $e_2 = 4 \times \frac{\pi\epsilon}{4} +  9 \times \frac{\pi\epsilon}{2} +  3 \times \frac{2\pi\epsilon}{3} + 3 \times \frac{3\pi\epsilon}{4} +  10 \times \pi\epsilon + 10 \times E$ \\ \hline

\hspace{-1mm}Grover$^{\{011,101\}}$\hspace{-1.5mm} 	& 4 & 31/12	& \hspace{-1.5mm}3250$\mu s$\hspace{-1.5mm}	& $e_1 = 6 \times 0.707107\epsilon + 3 \times 0.866025\epsilon + 10\times\epsilon + 6 \times 0.707107E + 6 \times E$ \\ 
&&&& $e_2 = 4 \times \frac{\pi\epsilon}{4} +  10 \times \frac{\pi\epsilon}{2} +  3 \times \frac{2\pi\epsilon}{3} + 2 \times \frac{3\pi\epsilon}{4} +  12 \times \pi\epsilon + 12 \times E$ \\ \hline

\hspace{-1mm}Grover$^{\{010,100\}}$\hspace{-1.5mm} 	& 4 & 32/12	& \hspace{-1.5mm}3280$\mu s$\hspace{-1.5mm}	& $e_1 = 6 \times 0.707107\epsilon + 3 \times 0.866025\epsilon + 10\times\epsilon + 6 \times 0.707107E + 6 \times E$ \\ 
&&&& $e_2 = 3 \times \frac{\pi\epsilon}{4} +  10 \times \frac{\pi\epsilon}{2} +  3 \times \frac{2\pi\epsilon}{3} + 3 \times \frac{3\pi\epsilon}{4} +  13 \times \pi\epsilon + 12 \times E$ \\ \hline

\hspace{-1mm}Grover$^{\{000,111\}}$\hspace{-1.5mm} 	& 4 & 31/13	& \hspace{-1.5mm}3492$\mu s$\hspace{-1.5mm}	& $e_1 = 4 \times 0.707107\epsilon + 5 \times 0.866025\epsilon + 10\times\epsilon + 6 \times 0.707107E + 7 \times E$ \\ 
&&&& $e_2 = 3 \times \frac{\pi\epsilon}{4} +  10 \times \frac{\pi\epsilon}{2} + 5 \times \frac{2\pi\epsilon}{3} + 1 \times \frac{3\pi\epsilon}{4} + 12 \times \pi\epsilon + 13 \times E$ \\ \hline

\hspace{-1mm}QFT4					& 4 & 13/6	& \hspace{-1.5mm}1582$\mu s$\hspace{-1.5mm}	& $e_1 = 1 \hspace{-0.5mm}\times\hspace{-0.5mm} 0.273262\epsilon + 1 \hspace{-0.5mm}\times\hspace{-0.5mm} 0.382683\epsilon + 2 \hspace{-0.5mm}\times\hspace{-0.5mm} 0.521005\epsilon + 1 \hspace{-0.5mm}\times\hspace{-0.5mm} 0.866025 \epsilon + 1 \hspace{-0.5mm}\times\hspace{-0.5mm} 0.92388\epsilon $ \\
&&&& $+2\hspace{-0.5mm}\times\hspace{-0.5mm} 0.980785\epsilon + 4\hspace{-0.5mm}\times\hspace{-0.5mm}\epsilon + 1\hspace{-0.5mm}\times\hspace{-0.5mm} 0.19509E + 2\hspace{-0.5mm}\times\hspace{-0.5mm} 0.382683E + 3\hspace{-0.5mm}\times\hspace{-0.5mm} 0.707107E$ \\

&&&& $e_2 = 1 \hspace{-0.5mm}\times\hspace{-0.5mm} \frac{3\pi\epsilon}{8} + 4 \hspace{-0.5mm}\times\hspace{-0.5mm} \frac{\pi\epsilon}{2} + 2 \hspace{-0.5mm}\times\hspace{-0.5mm} \frac{9\pi\epsilon}{16} + 1 \hspace{-0.5mm}\times\hspace{-0.5mm} \frac{2\pi\epsilon}{3} + 2 \hspace{-0.5mm}\times\hspace{-0.5mm} 0.825557\pi\epsilon + 1 \hspace{-0.5mm}\times\hspace{-0.5mm} \frac{7\pi\epsilon}{8}$ \\
&&&& $ + 1 \hspace{-0.5mm}\times\hspace{-0.5mm} 0.911896\pi\epsilon + 1 \hspace{-0.5mm}\times\hspace{-0.5mm} \pi\epsilon + 6 \times E$ \\ \hline

\hspace{-1mm}QFT5 					& 5 & 22/10	& \hspace{-1.5mm}2669$\mu s$\hspace{-1.5mm}	& $e_1 = 1 \hspace{-0.5mm}\times\hspace{-0.5mm}0.138284\epsilon + 1 \hspace{-0.5mm}\times\hspace{-0.5mm}0.273262\epsilon + 1 \hspace{-0.5mm}\times\hspace{-0.5mm} 0.521005\epsilon + 1 \hspace{-0.5mm}\times\hspace{-0.5mm} 0.55557\epsilon + 1\hspace{-0.5mm}\times\hspace{-0.5mm} 0.707107e$ \\
&&&& $ + 1\hspace{-0.5mm}\times\hspace{-0.5mm} 0.722528\epsilon \hspace{-0.5mm}+\hspace{-0.5mm} 1\hspace{-0.5mm}\times\hspace{-0.5mm} 0.831147\epsilon \hspace{-0.5mm}+\hspace{-0.5mm} 2\hspace{-0.5mm}\times\hspace{-0.5mm} 0.866025\epsilon \hspace{-0.5mm}+\hspace{-0.5mm} 1 \hspace{-0.5mm}\times\hspace{-0.5mm} 0.951173\epsilon \hspace{-0.5mm}+\hspace{-0.5mm} 2 \hspace{-0.5mm}\times\hspace{-0.5mm} 0.995185\epsilon$ \\
&&&& $ + 4\hspace{-0.5mm}\times\hspace{-0.5mm} \epsilon + 1\hspace{-0.5mm}\times\hspace{-0.5mm} 0.098017E + 2\hspace{-0.5mm}\times\hspace{-0.5mm} 0.19509E + 3\hspace{-0.5mm}\times\hspace{-0.5mm} 0.382683E + 4\hspace{-0.5mm}\times\hspace{-0.5mm} 0.707107E$ \\

&&&& $e_2 =  1 \hspace{-0.5mm}\times\hspace{-0.5mm} \frac{3\pi\epsilon}{16} + 2 \hspace{-0.5mm}\times\hspace{-0.5mm} \frac{15\pi\epsilon}{32} + 4 \hspace{-0.5mm}\times\hspace{-0.5mm} \frac{\pi\epsilon}{2} + 1 \hspace{-0.5mm}\times\hspace{-0.5mm} 0.59988\pi\epsilon + 2 \hspace{-0.5mm}\times\hspace{-0.5mm} \frac{2\pi\epsilon}{3} +  1 \hspace{-0.5mm}\times\hspace{-0.5mm} \frac{11\pi\epsilon}{16}$ \\
&&&& $+ 1 \hspace{-0.5mm}\times\hspace{-0.5mm} 0.74298\pi\epsilon + 1 \hspace{-0.5mm}\times\hspace{-0.5mm} \frac{3\pi\epsilon}{4} +  1 \hspace{-0.5mm}\times\hspace{-0.5mm} 0.825557\pi\epsilon + 1 \hspace{-0.5mm}\times\hspace{-0.5mm} 0.911896\pi\epsilon$ \\
&&&& $ + 1 \hspace{-0.5mm}\times\hspace{-0.5mm} 0.955841\pi\epsilon + 6 \hspace{-0.5mm}\times\hspace{-0.5mm} \pi\epsilon + 10 \times E$ \\ \hline

\hspace{-1mm}Toffoli-4 				& 5 & 21/11	& \hspace{-1.5mm}2832$\mu s$\hspace{-1.5mm}	& $e_1 = 8 \times 0.707107\epsilon + 2 \times 0.866025\epsilon + 4\times\epsilon + 3 \times 0.707107E + 8 \times E$ \\
&&&& $e_2 = 8 \times \frac{\pi\epsilon}{4} +  4 \times \frac{\pi\epsilon}{2} +  2 \times \frac{2\pi\epsilon}{3} + 7 \times \pi\epsilon + 11 \times E$ \\ \hline
\end{tabular}
}
\caption{Benchmark circuits. Errors shown per both definitions of the single- and two-qubit $R/XX$ gate error models, as described in (\ref{eqn:ed},\ref{eqn:2ed}).}
\label{tab:benchmarks}
\end{table*}

Table \ref{tab:benchmarks} reports the result of the application of the above techniques to the design and optimization of circuits implementing advanced quantum computational experiments.  None of the implementations proposed in Table \ref{tab:benchmarks} have yet been demonstrated in an experiment.  Table \ref{tab:benchmarks} lists the name of the algorithm/function, the number of qubits the developed implementation uses, the number of physical-level single-qubit/two-qubit pulses, the overall runtime of the circuit, and all sources of errors.  All implementations reported in Table \ref{tab:benchmarks} are true implementations of the respective algorithms/functions, in that we report exact unitaries (VS those up to a relative phase), treat black boxes as black boxes (no optimizations crossing black box boundaries), as well as precisely follow the formulation of the respective algorithms and specifications. Our computations are thus properly scalable to accept larger numbers of qubits.  The quantum state is furthermore initialized to $\ket{00000}$ (as opposed to a state containing partial results of a computation) before any of the circuits are applied.  

For Grover's algorithm, we considered the scenario when the three-bit Boolean function $f(x)=f(x_1,x_2,x_3)$ implemented as Grover's oracle, $\ket{x,y} \mapsto \ket{x, y \oplus f(x)}$, marks some two items in the database.  There are 28 such functions.  The superscript in the name of the Grover's function in Table \ref{tab:benchmarks} indicates the bit strings encoded by the oracle the search over which is being reported.  The efficiency of the Grover's circuit is determined by the efficiency of the implementation of the oracle, which in turn is determined by the Hamming distance between those items it marks.  We treat the oracle as a black box, and do not allow optimizations across the boundary of the black box.  While we can implement Grover's algorithm over any oracle marking two items, only a few representative circuits are actually included in the Table.  To implement Grover's algorithm over a Boolean function marking one item, the oracle needs to be a Toffoli-4 gate (the triply-controlled Toffoli, $\textsc{TOF}[a,b,c;d]$).  We calculated that the single Grover's iteration requires 16 two-qubit gates.  The algorithmic probability of reading out the correct answer upon applying a single Grover's iterate is $0.78125$, which beats the classical probability of finding the answer with the single query, $0.125$, therefore, it may be interesting to run such an experiment, as well. 

Grover's algorithm may furthermore be implemented using the oracle function $g$ computing the unknown Boolean function $f$ into phase, $g:\ket{x} \mapsto (-1)^{f(x)}\ket{x}$ \cite{bk:nc}.  For the 2-out-of-8 search the function $g$ can be thought of as a $Z/CZ$ circuit---see Lemma \ref{lem:2} outlining the construction of efficient $Z/CZ$ circuits.  The number of $CZ$ gates used ranges between one and three.  This means that the entire Grover's search with such a phase oracle can be implemented using between $6$ and $8$ two-qubit $XX$ gates.  A 1-out-of-8 search with the single Grover's iterate and algorithmic success probability of $0.78125$ over phase oracle can be accomplished with $10$ two-qubit $XX$ gates.  Finally, some 4-qubit Grover's searches may be possible to demonstrate.  Specifically, the search for phase items $\{1110,1111\}$ can be accomplished by a trapped ions circuit with $16$ two-qubit $XX$ gates over algorithmic success probability of $0.78125$.

The QFT5 circuit reported in Table \ref{tab:benchmarks} can be compared head-to-head to the one found in \cite{ar:deb}.  Specifically, both circuits are true \cite[Figure 5.1]{bk:nc} (as opposed to semiclassical, \cite{quant-ph/9511007}) implementations of the 5-qubit QFT, featuring 10 two-qubit gates.  Our circuit benefited from careful design, and as a result features the single-qubit gate count of just $22$ compared to $70$ in \cite{ar:deb}.  This constitutes the reduction of the single-qubit gate count by a factor of more than three, clearly illustrating the benefits of our approach.  

To our knowledge, the Toffoli-4 gate circuit reported in Table \ref{tab:benchmarks} is the first such containing no more than 11 two-qubit gates.  Previous best result is 12 CNOTs \cite{ar:m}. 

\section{Conclusion}
In this paper we reported a complete strategy for automatic execution of quantum algorithms on a trapped ions quantum machine.  Our contributions include the design of the complete data flow from the algorithm level down to the physical level, algorithms for circuit cost optimization---specifically, due to combining decompositions of gates such as to enforce gate cancellations and single-qubit gate optimization, and physical-level designs of quantum logical gates and computational primitives.  In particular, we demonstrated simplified designs of computational primitives suitable for a range of QIP proposals relying on the physical control provided by the $R$ and $XX$ gates, and designed optimized computational experiments suitable for the execution on a specific machine available in the lab.  We furthermore note that the circuit optimization techniques developed in this paper are basic and generic---the key limitation is the reliance on the specific gate library, and thus can be modified and applied in conjunction with numerous optimization criteria, or even a over different QIP platform. 

Our results help to bridge the gap between quantum computational experiments and a fully-fledged quantum computer: indeed, our approach allows to automatically design and execute quantum algorithms on the existing hardware, which may be described as programming a quantum computer.

\section{Acknowledgements}

I thank Prof. Christopher Monroe from the University of Maryland--College Park for discussions and help in the preparation of this manuscript.

Circuit diagrams were drawn using qcircuit.tex package, \href{http://physics.unm.edu/CQuIC/Qcircuit/}{http://physics.unm.edu/CQuIC/Qcircuit/}.

This material was based on work supported by the National Science Foundation, while working at the Foundation. Any opinion, finding, and conclusions or recommendations expressed in this material are those of the author and do not necessarily reflect the views of the National Science Foundation.


\end{document}